\documentclass[
final
]{dmtcs-episciences}

\usepackage{subfigure}
\usepackage{amsmath,amsthm,stmaryrd}
\usepackage{multirow,bigdelim}
\usepackage{tikz}

\usepackage[round]{natbib}

\newtheorem{theorem}{Theorem}
\newtheorem{conjecture}[theorem]{Conjecture}
\newtheorem{proposition}[theorem]{Proposition}
\newtheorem{lemma}[theorem]{Lemma}
\theoremstyle{definition}
\newtheorem{fact}{Fact}

\newcommand{\Line}[1]{\overline{#1}}

\newcommand{\dbelong}{de Bruijn-Erd\H{o}s property}

\author{Laurent Beaudou\affiliationmark{1}\thanks{supported by ANR \textsc{DISTANCIA} (ANR-17-CE40-0015) and \textsc{GraphEn} (ANR-15-CE40-0009).}
  \and Giacomo Kahn\affiliationmark{2}\thanks{supported by  the European Union's {\em Fonds
  Europ\'een de D\'eveloppement R\'egional (\textsc{feder})} program
though project AAP ressourcement S3-DIS4 (2015-2018).}
  \and Matthieu Rosenfeld\affiliationmark{3}\thanks{supported by ANR \textsc{GraphEn} (ANR-15-CE40-0009).}}
\title{Bisplit graphs satisfy the Chen-Chv\'atal conjecture}
\affiliation{
  Higher School of Economics, Moscow, Russian Federation\\
  Universit\'e d'Orl\'eans, France\\
  Universit\'e de Li\`ege, Belgium}
\received{2018-9-10}

\revised{2019-3-8}

\accepted{2019-4-22}

\begin{document}
\publicationdetails{21}{2019}{1}{5}{4813}
\maketitle
\begin{abstract}
  In this paper, we give a lengthy proof of a small result! A graph is
  bisplit if its vertex set can be partitioned into three stable sets
  with two of them inducing a complete bipartite graph. We prove that
  these graphs satisfy the Chen-Chv\'atal conjecture: their metric
  space (in the usual sense) has a universal line (in an unusual
  sense) or at least as many lines as the number of vertices.
\end{abstract}


Given a set of $n$ points in the Euclidean plane, they are all
collinear or they define at least $n$ distinct lines. This result is a
corollary of Sylvester-Gallai Theorem (suggested by \cite{sylvester}
in the late nineteenth century and proven by Gallai forty years later
as reported by \cite{erdos82}). Later,~\cite{debruijnerdos1948} proved
a theorem on collections of subsets, which also implies that $n$
points are either collinear or define at least $n$ distinct lines.

The notion of line admits several generalizations, one of which is of
interest for us in this paper. Namely, given a metric space
$(X,\rho)$, we say that an element $b$ in $X$ is {\em between}
elements $a$ and $c$ if $\rho(a,b) + \rho(b,c) = \rho(a,c)$. More
generally, we say that three elements of $X$ are {\em collinear} if
one of them is between the other two. In that setting, {\em the line
  generated by $a$ and $b$} (denoted $\Line{ab}$) is the set $\{a,b\}$
completed by all elements collinear with $a$ and $b$. Ten years
ago,~\cite[Question~1]{chenchvatal2008} asked what has now become, by
lack of counter-example, the Chen-Chv\'atal Conjecture.

\begin{conjecture}[\cite{chenchvatal2008}]
  \label{conjecturecc}
  Every finite metric space $(X,\rho)$ where no line consists of the
  entire ground set $X$ determines at least $|X|$ distinct lines.
\end{conjecture}

A line consisting of the entire ground set is called a {\em universal
  line}. Conjecture~\ref{conjecturecc} remains unsettled when
restricted to graph metrics (for connected graphs). Let us say that a
graph $G$ on $n$ vertices {\em has the \dbelong} if the metric space
induced by $G$ has a universal line or at least $n$ distinct lines. In
a paper gathering more coauthors than pages, Beaudou, Bondy, Chen,
Chiniforooshan, Chudnovsky, Chv\'atal, Fraiman and
Zwols~\cite[Theorem~1]{beaudou_chordal} proved that connected chordal
graphs have the \dbelong. Recently, \cite{aboulker_matamala} improved
this result by encompassing a larger family of graphs.

One may find out quite easily that connected graphs with a bridge have
a universal line. As noted in~\cite[Section 3]{beaudou_chordal}
connected bipartite graphs also have the \dbelong~(each line generated
by both ends of an edge is universal).

A significant number of results have appeared concerning the
asymptotic number of lines in a graph with no universal lines. A
notable one is due to Aboulker, Chen, Huzhang, Kapadia and Supko. They
prove~\cite[Theorem~7.4]{aboulker_chen} that graphs with $n$ vertices
and diameter $d(n)$ have $\Omega((n/d(n))^{4/3})$ distinct lines or a
universal line. This implies that any class of graphs with bounded
diameter ultimately has the \dbelong.

Thus, large graphs of diameter 2 have the \dbelong. Chv\'atal filled the
gap for small graphs of diameter 2 by proving the stronger
result~\cite[Theorem~1]{chvatal2014} that every 1-2 metric space has
the \dbelong.

A connected graph $G$ is {\em bisplit} if its vertex set can be
partitioned into three stable sets $X$, $Y$ and $Z$ such that $Y$ and
$Z$ induce a complete bipartite graph. This class of graphs has
diameter bounded by 4. Thus they ultimately have the
\dbelong. Moreover, bisplit graphs are one step away from bipartite
graphs (when $Z$ or $Y$ is empty). One may think that they could
easily be tamed. It turns out that we could not find a short proof.

In this paper, we prove that bisplit graphs have the \dbelong.

\begin{theorem}
  \label{thm:main}
  For any integer $n$ greater than or equal to 2, all connected
  bisplit graphs on $n$ vertices have a universal line or at
  least $n$ distinct lines.
\end{theorem}

This settles Problem~1 from \cite{chvatal2018}. Chen and
Chiniforooshan (see final note of~\cite{chvatal2018} in the online
version) provided a proof using computer enumeration for small
cases. Our proof does not use computer enumeration.


\section{Calculus 101}
\label{sec:calculus}

In this section, we state easy results that will be used in the flow
of the proof of Theorem~\ref{thm:main}. We do not give the proof of
the following lemma. It is straightforward.

\begin{lemma}
  For any integer $x$, ${x \choose 2} \geq x - 1$. Besides, for any
  pair of positive integers $x$ and $y$, $xy \geq x+y-1$. Moreover, if
  both $x$ and $y$ are greater than or equal to 2, then $xy \geq x +
  y$.
\end{lemma}

The next lemma is a bit more tedious. While one might use any computer
to have this answer, we give a formal proof hereafter. The reader is
advised to skip the proof if the word {\em trinomial} does not sound
thrilling enough.

\begin{lemma}
  \label{lem:chiant}
  Given two positive integers $x$ and $y$,
  \begin{equation}
    {y \choose 2} + {{\lceil \frac{2x}{y} \rceil} \choose 2} < x + y - 1
    \label{eq:lem}
  \end{equation}
  if and only if $y = 2$ and $x$ is in $\{1,2\}$ or $y = 3$ and $x=3$.
\end{lemma}

\begin{proof}
  It is easy to check that the solutions provided satisfy the
  inequality. Let us now assume that we are given two integers $x$ and
  $y$ such that $x$ is positive, $y$ is at least 2 and they satisfy
  the inequality.
  
  Note that
  \begin{equation*}
    {{\lceil \frac{2x}{y} \rceil} \choose 2} = \frac{ \lceil \frac{2x}{y} \rceil^2 - \lceil \frac{2x}{y} \rceil}{2} \geq \frac{\frac{4x^2}{y^2} - \frac{2x}{y}}{2}.
  \end{equation*}
  Inequality~\eqref{eq:lem} implies then that
  \begin{equation*}
    \frac{y^2 - y}{2} + \frac{\frac{4x^2}{y^2} - \frac{2x}{y}}{2} - x - y + 1 < 0.
  \end{equation*}
  This inequality can be simplified and expressed as a trinomial on variable $x$.
  \begin{equation}
    \label{eq:trinom}
    4x^2 - 2 (y ^2 + y) x + (y^4 -3y^3+2y^2) < 0.
  \end{equation}
  
  Since the coefficient of $x^2$ is positive, \eqref{eq:trinom} has
  solutions if and only if its discriminant is positive. This
  discriminant can be simplified and \eqref{eq:trinom} has solutions
  if and only if,
  \begin{equation*}
    -3y^2 + 14y - 7 \geq 0.
  \end{equation*}
  This trinomial is even easier than the first one. The value of $y$
  must be between $1$ and $4$. We may now check each case
  individually.
  \begin{itemize}
    \item If $y = 1$, Inequality~\eqref{eq:lem} becomes ${2x \choose 2} <
      x$ which has no integer solution.
    \item If $y = 2$, Inequality~\eqref{eq:lem} reads ${x \choose 2} < x$ and the only solutions
      are for $x$ in $\{1,2\}$.
    \item If $y = 3$, we may check the first values of $x$ and see
      that the only solution is for $x = 3$.
    \item If $y = 4$, the solutions of the trinomial~\eqref{eq:trinom}
      are in the interval $[3,7]$. One may check that no integral
      solution exists.
  \end{itemize}
  This concludes our proof of Lemma~\ref{lem:chiant}.
\end{proof}


\section{Proof of Theorem~\ref{thm:main}}

 Let us consider a bisplit graph $G$. Among all valid partitions for
 $G$, we consider one that maximizes the size of $Y \cup Z$. Then, we
 specify a more precise partition of vertices:
\begin{itemize}
\item the set $X$ is split into three sets: the set $X_Y$
  (respectively $X_Z$) made of vertices of $X$ whose only neighbours
  are in $Y$ (respectively $Z$) and the set $X'$ made of vertices with
  neighbours in both $Y$ and $Z$,
\item the set $Y$ (respectively $Z$) is split into two sets: the set
  $Y_1$ (respectively $Z_1$) made of vertices with at least one
  neighbour in $X_Y$ (respectively $X_Z$) and the set $Y_2$
  (respectively $Z_2$) made of vertices with no neighbour in $X_Y$
  (respectively $X_Z$).
\end{itemize}

Armed with this new partition, we may derive a sketch of all possible
distances between vertices of a bisplit graph. Numbers on
Figure~\ref{fig:general} refer to the possible distance between two
vertices in a same set (numbers in the circles), or between two
vertices in two separate sets (numbers on edges between two sets). For
example, if $a$ and $b$ are two vertices in $X'$, they both have
neighbours in $Y$ and $Z$. If they have a common neighbour they are at
distance 2 and if not, they must be at distance 3. 


\begin{figure}[ht]
  \scriptsize
    \begin{center}
    \begin{tikzpicture}[scale=1.5]
      \draw (0,2) circle(.5);
      \draw (-3,0) circle(.5);
      \draw (3,0) circle(.5);
      \draw[dashed] (-1,0) circle(.35);
      \draw[dashed] (1,0) circle(.35);
      \draw[dashed] (-1,-1) circle(.35);
      \draw[dashed] (1,-1) circle(.35);
      \draw[rounded corners=3pt] (-1.6,-1.6) rectangle (-.4,.6);
      \draw[rounded corners=3pt] (.4,-1.6) rectangle (1.6,.6);
      
      \draw (-.4,-.5) -- (.4,-.5);
      \draw (-2.5,0) -- (-1.35,0);
      \draw (2.5,0) -- (1.35,0);
      \draw (-2.6,-.3) -- (-1.35,-1);
      \draw (2.6,-.3) -- (1.35,-1);
      \draw (-.29,1.6) -- (-1,.6) node[sloped, midway, above] {$1,2$};
      \draw (.29,1.6) -- (1,.6) node[sloped, midway, above] {$1,2$};
      \draw (-.45,1.8) -- (-2.7,.4) node[sloped, midway, above] {$2,3$};
      \draw (.45,1.8) -- (2.7,.4) node[sloped, midway, above] {$2,3$};
      \draw plot [smooth] coordinates { (-2.7,-0.4) (-2,-1.8)  (.42,-1.57)};
      \draw plot [smooth] coordinates { (2.7,-0.4) (2,-1.8)  (-.42,-1.57)};
      \draw plot [smooth] coordinates { (-3,-0.5) (-2.4,-2) (0,-2.3) (2.4,-2) (3,-.5)};

      \draw (0,-2.3) node[below] {$3$};
      \draw (-2,-0.5) node {$3$};
      \draw (-2,0) node[above] {$1,3$};
      \draw (2,-0.5) node {$3$};
      \draw (2,0) node[above] {$1,3$};
      \draw (0,-0.5) node[above] {$1$};
      \draw (-2,-1.65) node {$2$};
      \draw (2,-1.65) node {$2$};
      \draw (-1.3,-1.5) node[left] {$Y$}; 
      \draw (1.3,-1.5) node[right] {$Z$}; 
      \draw (0,2) node [below] {$2,3$} node [above] {$X'$}; 
      \draw (-1,0) node  {$Y_1$}; 
      \draw (1,0) node  {$Z_1$};
      \draw (-1,-1) node  {$Y_2$}; 
      \draw (1,-1) node  {$Z_2$};
      \draw (-3,0) node [above] {$X_Y$} node [below] {$2,4$}; 
      \draw (3,0) node [above] {$X_Z$} node [below] {$2,4$};
      \draw (-1,-.5) node {$2$};
      \draw (1,-.5) node {$2$};
    \end{tikzpicture}
    \end{center}
    \caption{Possible distances in a bisplit graph}
    \label{fig:general}
  \end{figure}
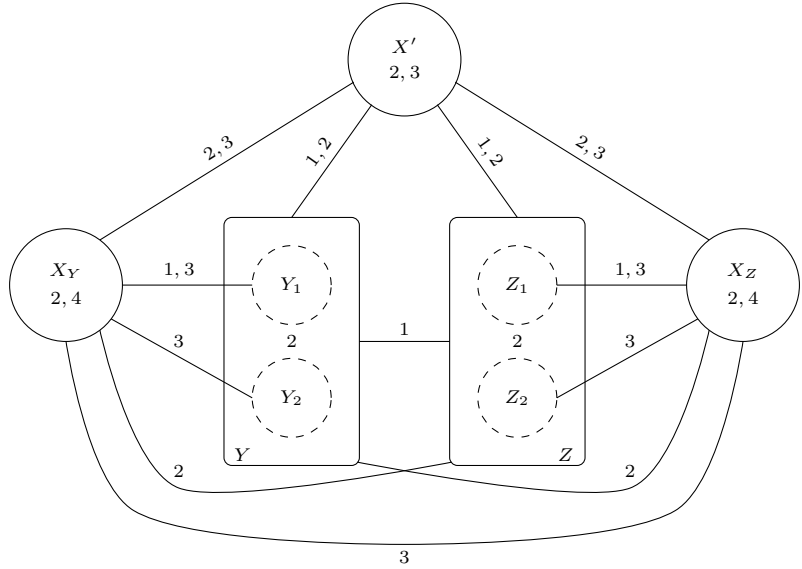


A vertex in $X_Y$ adjacent to all the set $Y$ could be put in
the set $Z$ from the start. Since we chose to maximize the size of $Y
\cup Z$ among all valid partitions, we may assume that,
\begin{equation}
  \text{no vertex in }X_Y  \text{ is complete to } Y.\label{eq:nofull}
\end{equation}

Moreover, if either $X'$, $Y$ or $Z$ is empty, the whole graph is
bipartite and thus satisfies the \dbelong. From now on, we consider
that
\begin{equation}
  X', Y \text{ and } Z \text{ are not empty.} \label{eq:X'}
\end{equation}

Before diving into the proof, we warn the reader that we will
extensively use a simple fact. For three points to be collinear in a
metric space with distances ranging from 0 to 4, the triple of
distances they define can only be one of the following: $(1,1,2),
(1,2,3), (1,3,4)$ or $(2,2,4)$. They are the only cases when the
triangle inequality is tight.

The proof is declined as a case study depending on the emptiness of
sets $X_Y$ and $X_Z$.

\subsection{When both $X_Y$ and $X_Z$ are empty}

Restricting to the case when $X_Y$ and $X_Z$ are empty amounts to
considering bisplit graphs for which all vertices in $X$ have at least
one neighbour in $Y$ and one in $Z$.

\begin{proposition}
  If a bisplit graph is such that all vertices in $X$ have at least
  one neighbour in both $Y$ and $Z$, then it has the \dbelong.
  \label{prop:non-deg}
\end{proposition}

\begin{proof}
  Let $G$ be such a bisplit graph. By \eqref{eq:X'}, sets $X$, $Y$ and
  $Z$ are not empty. If the distance between any two vertices is at
  most 2, then $G$ has the de Bruijn-Erd\H{o}s property (Chv\'{a}tal
  proved that every 1-2 metric space has the \dbelong{}~\cite[Theorem
    1]{chvatal2014}). So there are at least two vertices $a$ and $b$
  at distance 3 from each other. Both $a$ and $b$ must be vertices in
  $X$ since all distances involving some vertex in $Y$ or $Z$ is 1 or
  2 (see Figure~\ref{fig:nondeg}). Moreover, $Y$ and $Z$ both have
  cardinality at least 2 (if $Y$ was a singleton, $a$ and $b$ would
  have a common neighbour since they must have a neighbour in
  $Y$). Finally, vertex $a$ cannot be complete neither to $Y$ nor to
  $Z$ (or it would be at distance 2 from $b$).


  \begin{figure}[ht]
    \begin{center}
    \begin{tikzpicture}
      \draw (0,0) circle(1);
      \draw (-2,-3) circle(1);
      \draw (2,-3) circle(1);
      \draw (-1,-3) -- (1,-3);
      \draw (236.3:1) -- (236.3:2.605);
      \draw (-56.3:1) -- (-56.3:2.605);

      \draw (-1,-1.5) node [left] {$1,2$};
      \draw (1,-1.5) node [right] {$1,2$};
      \draw (0,-3) node [above] {$1$};
      \draw (0,0) node [below] {$2,3$} node [above] {$X$};
      \draw (-2,-3) node [below] {$2$} node [above] {$Y$};
      \draw (2,-3) node [below] {$2$} node [above] {$Z$};
    \end{tikzpicture}
    \end{center}
    \caption{Possible distances in a bisplit graph where $X_Y$ and $X_Z$ are empty}
    \label{fig:nondeg}
  \end{figure}
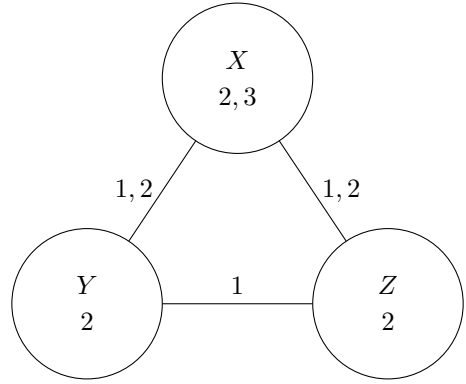

  Concerning notation, we shall use $N(v)$ to represent the
  neighbourhood of a vertex $a$. If we want to restrict ourselves to
  the neighbours of a vertex $a$ in a set $S$, we shall write
  $N_S(a)$. Moreover, for an integer $i$, $N^i(a)$ denotes the
  vertices which are at distance exactly $i$ from $a$. Similarly, we
  may restrict to a specific set by using a subscript. Finally,
  $\overline{N_S(a)}$ denotes the non-neighbours of vertex $a$ within
  set $S$.
  
  Let $F_X$ denote the set of all lines generated by $a$ and another
  vertex $x$ in $X \setminus \{a\}$. For every such vertex $x$, the
  intersection of $\Line{ax}$ with $X$ is always restricted to the
  generators $a$ and $x$ (all tight triple of distances must have a
  $1$ or a $4$). Then $F_X$ is a set of $|X|-1$ distinct lines.
  \begin{equation}
    \mbox{For all } x \mbox{ in } X \setminus \{a\}\mbox{, } \Line{ax}
    = \{a,x\} \cup \begin{cases} N(a) \cap N(x) & \mbox{ if } ax =
      2\\ N(a) \cup N(x) & \mbox{ if } ax = 3.\end{cases}
    \tag{$L_X$}\label{eq:LX}
  \end{equation}
  
  Let $F_Y$ denote the set of all lines generated by $a$ and a vertex
  in $Y$ which is not adjacent to $a$.
  \begin{equation}
    \mbox{For all } y \mbox{ in } \overline{N_Y(a)} \mbox{, } \Line{ay} = \{a\}
    \cup (N_X^3(a) \cap N_X(y)) \cup \{y\} \cup N_Z(a).
    \tag{$L_Y$}\label{eq:LY}
  \end{equation}

  Since the intersection of such a line with $Y$ is restricted to the
  singleton containing the other generator, all those lines must be
  distinct.

  Let $F_Y'$ be the set of all lines generated by a pair of vertices
  in $Y$ which are adjacent to $a$.
  \begin{equation}
    \mbox{For all } y \text{ and } y' \mbox{ in } N_Y(a) \mbox{, }
    \Line{yy'} = N_X(y) \cap N_X(y') \cup \{y,y'\} \cup Z.
    \tag{$L'_Y$}\label{eq:LY'}
  \end{equation}

  Since the intersection of such a line with $Y$ is exactly the pair
  of generators, all those lines are also distinct. Moreover, they are
  distinct from lines in $F_Y$ since the latter intersect $Y$ on a
  singleton. As a consequence, $F_Y \cup F_Y'$ is a set of 
  \begin{equation*}
    {d_Y(a) \choose 2} + |Y| - d_Y(a) 
  \end{equation*}
  distinct lines. Note that this quantity is always greater than or
  equal to $|Y| - 1$.

  We define $F_Z$ and $F'_Z$ similarly.
  
  \paragraph{No intersection.}
  We first prove that $F_X$ does not intersect the other families of
  lines. For a contradiction, suppose that $F_X$ intersects
  $F_Y$. Then there are two vertices $x$ in $X \setminus \{a\}$ and
  $y$ in $Y \cap \overline{N(a)}$ such that $\Line{ax} = \Line{ay}$.
  Let us focus on the intersection with $Y$. It must be exactly
  $\{y\}$. If $d(a,x) = 2$ then \mbox{$\{y\} = N(a) \cap N(x)$} which is
  impossible since $ay$ is not an edge in $G$. If $d(a,x)=3$ then
  $\{y\} = N(a) \cup N(x)$ but $a$ and $x$ have no common neighbour
  and at least one neighbour each in $Y$. This is also a
  contradiction. Now suppose that $F_X$ intersects $F'_Y$. Then there
  are three vertices $x$ in $X \setminus \{a\}$ and $y$ and $y'$ in
  $N_Y(a)$ such that $\Line{ax} = \Line{yy'}$. This line must contain
  the whole set $Z$. Since $a$ is not complete to $Z$, $N(a) \cap
  N(x)$ cannot contain $Z$. From this and statement~\eqref{eq:LX}, we
  derive that $a$ and $x$ must be at distance 3. But then $\{y,y'\}$
  must be the union of $N_Y(a)$ and $N_Y(x)$. Since both $y$ and $y'$
  are neighbours of $a$ and as $x$ must have a neighbour in $Y$ (by
  our initial hypothesis), $a$ and $x$ must have a common neighbour
  which is a contradiction.

  Let us now prove that $F_Y$ does not intersect $F_Z$. For a
  contradiction, assume that there are two vertices $y$ in $Y \cap
  \overline{N(a)}$ and $z$ in $Z \cap \overline{N(a)}$ such that
  $\Line{ay} = \Line{az}$. By looking at the intersection with $Z$,
  this would mean that $z$ is in $N(a)$ which is a contradiction. We
  keep going and prove that $F_Y$ does not intersect $F'_Z$. Lines in
  $F'_Z$ contain the whole set $Y$ but lines in $F_Y$ contain only
  non-neighbour of $a$. Since $a$ has at least a neighbour in $Y$
  these lines cannot be equal.

  Finally, let us prove that $F'_Y$ does not intersect $F_Z'$. Once
  again, if two such lines were equal, they would contain the whole
  sets $Y$ and $Z$ which should be of cardinality 2 but then $a$ would
  be complete to both of them which is impossible.

  In the end, we may sum all those lines together. We obtain at least
  \begin{equation*}
    |X| - 1 + |Y| - 1 + |Z| - 1
  \end{equation*}
  distinct lines. They all contains vertex $a$.

  \paragraph{Reaching for the last three lines.}
  If $X$ has cardinality 4 or more, we may consider all the lines
  generated by a pair of vertices in $X \setminus \{a\}$. Those line
  do not contain $a$ and are distinct from each other. There are at
  least three such lines.

  If $X$ has cardinality at most 3, any line generated by a pair of
  vertices in $X$ contains $a$ or $b$ (recall that $b$ is a vertex at
  distance 3 from $a$). We shall distinguish two extra lines. Vertex
  $b$ must have a neighbour $y_b$ in $Y$ and a neighbour $z_b$ in $Z$
  by our hypothesis. Similarly, vertex $a$ has a neighbour $y_a$ in
  $Y$ and $z_a$ in $Z$. Since $a$ and $b$ are at distance 3, those
  four vertices must be distinct.

  Now if $X$ has cardinality 2, the line $\Line{y_az_b}$ is
  universal. And if $X$ has cardinality 3, the third vertex and $b$
  generate one line that does not go through $a$. Moreover we may
  consider lines $\Line{y_ay_b}$ and $\Line{z_az_b}$. These lines do
  not contain neither $a$ nor $b$. Thus they are different from all
  the lines described above. The only issue comes if they are
  equal. In that case, $Y$ and $Z$ must have cardinality 2 and we know
  almost everything about the graph. Let $c$ be the third vertex in
  $X$. It must have at least one neighbour in $Y$. Without loss of
  generality, we may assume it is $y_a$. Then for $\Line{y_az_b}$ not
  to be universal, $c$ must be a neighbour of $z_b$. For the same
  reason, $c$ is either adjacent to both $y_b$ and $z_a$ or to none of
  them. This leads to two graphs (see Figure~\ref{fig:twographs}).


  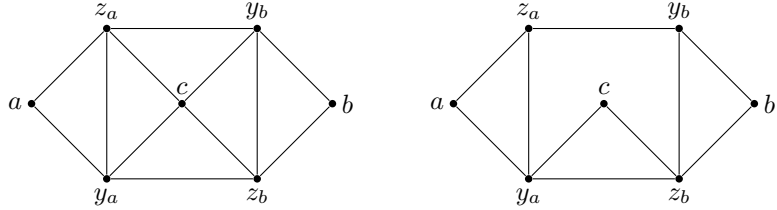
\begin{figure}[ht]
    \begin{center}
    \tikzstyle{vertex}=[circle,fill,black,inner sep=1pt]
    \begin{tikzpicture}
      \node[vertex] (c) at (0,0) {};
      \node[above] at (c) {$c$};
      \node[vertex] (za) at (-1,1) {};
      \node[above] at (za) {$z_a$}; 
      \node[vertex] (zb) at (1,-1) {};
      \node[below] at (zb) {$z_b$};
      \node[vertex] (a) at (-2,0) {};
      \node[left] at (a) {$a$};
      \node[vertex] (b) at (2,0) {}; 
      \node[right] at (b) {$b$};
      \node[vertex] (ya) at (-1,-1) {};
      \node[below] at (ya) {$y_a$};
      \node[vertex] (yb) at (1,1) {};
      \node[above] at (yb) {$y_b$};
      \draw (a) -- (za) -- (yb) -- (b) -- (zb) -- (ya) -- (a);
      \draw (c) -- (za) -- (ya) -- (c) -- (zb) -- (yb) -- (c);
    \end{tikzpicture}
    \qquad 
    \begin{tikzpicture}
      \node[vertex] (c) at (0,0) {};
      \node[above] at (c) {$c$};
      \node[vertex] (za) at (-1,1) {};
      \node[above] at (za) {$z_a$}; 
      \node[vertex] (zb) at (1,-1) {};
      \node[below] at (zb) {$z_b$};
      \node[vertex] (a) at (-2,0) {};
      \node[left] at (a) {$a$};
      \node[vertex] (b) at (2,0) {}; 
      \node[right] at (b) {$b$};
      \node[vertex] (ya) at (-1,-1) {};
      \node[below] at (ya) {$y_a$};
      \node[vertex] (yb) at (1,1) {};
      \node[above] at (yb) {$y_b$};
      \draw (a) -- (za) -- (yb) -- (b) -- (zb) -- (ya) -- (a);
      \draw (za) -- (ya) -- (c) -- (zb) -- (yb);
    \end{tikzpicture}
    \caption{The last two suspects.}
    \label{fig:twographs}
    \end{center}
  \end{figure}


  It is then straightforward to check that these two graphs have more
  than seven lines. This concludes our proof of Proposition~\ref{prop:non-deg}.
\end{proof}

\subsection{When $X_Y$ or $X_Z$ is non-empty}
In Proposition~\ref{prop:non-deg}, we proved that when $X_Y$ and $X_Z$
are empty, the graph satisfies the \dbelong. For the other cases, we
may assume without loss of generality that
\begin{equation}
  X_Y \text{ is not empty.} \label{eq:XY}
\end{equation}

As a direct consequence of~\eqref{eq:XY}, we may assume that 
\begin{equation}
  Y_1 \text{ has cardinality at least 2}
  \label{eq:Y1}
\end{equation}
because otherwise, the graph is not connected (if $Y_1$ is empty) or
has a bridge (if $Y_1$ is a singleton) which generates a universal
line. Moreover, if $Y$ has cardinality exactly 2, since every vertex
in $X_Y$ has degree at least 2, it is complete to $Y$ which
contradicts~\eqref{eq:nofull}. Therefore, we may consider that
\begin{equation}
  Y \text{ has cardinality at least } 3.
\label{eq:Y3}
\end{equation}

The remainder of our proof relies on a careful choice of families of
distinct lines. For any two sets $A$ and $B$ of vertices, we define
$F_{AB}$ to be the set of lines generated by any two vertices $a$ in
$A$ and $b$ in $B$.

\begin{fact}
The set $F_{X'X'}$ is made of ${|X'| \choose 2}$ distinct lines each
of which intersect $X'$ on exactly two vertices. Moreover, those lines
do not intersect neither $X_Y$ nor $X_Z$.
\label{fact:X'}
\end{fact}

\begin{proof} Fact~\ref{fact:X'} is obtained by checking the possible distances in
the graph (see Figure~\ref{fig:general}).
\end{proof}

\begin{fact}
The set $F_{Y_1Y}$ is made of ${|Y_1| \choose 2} + |Y_1||Y_2|$
distinct lines each of which intersects $Y$ on exactly two vertices
(the generators). Moreover, those lines do intersect $X_Y$ while they
do not intersect $X_Z$.
\label{fact:Y}
\end{fact}

\begin{proof}
Fact~\ref{fact:Y} is also obtained through a straightforward analysis
of Figure~\ref{fig:general}. Moreover, by looking at the intersection
with $X_Y$, we may derive from Facts~\ref{fact:X'} and~\ref{fact:Y}
that $F_{X'X'}$ and $F_{Y_1Y}$ are disjoint.
\end{proof}

\begin{fact}
  The set $F_{X_YZ}$ contains at least $|Z|$ lines which intersect $Z$
  on a singleton (the generator) and which contain the whole set
  $Y$. Moreover the intersection of such lines with $X_Y$ is not
  empty.
  \label{fact:XZ}
\end{fact}

\begin{proof}
  Fact~\ref{fact:XZ} is obtained by pinning one vertex $x$ in $X_Y$ and
studying lines $\Line{xz}$ for every vertex $z$ in $Z$. The
intersection of these lines with $X_Y$ guarantees that they are
distinct from lines in $F_{X'X'}$. By \eqref{eq:Y3} and Fact~\ref{fact:Y},
they are also distinct from lines in $F_{Y_1Y}$.
\end{proof}

\begin{fact}
  \label{fact:YZ}
  If $G$ has no universal line, and $X_Y$ is not empty, then $F_{YZ}$
  has cardinality at least 2. Moreover, every line in $F_{YZ}$
  contains all $X_Y, Y$ and $Z$.
\end{fact}

\begin{proof}
  Let $z$ be a vertex in $Z$ (it is not empty otherwise $G$ is
  bipartite). For any vertex $y$ in $Y$, the line $\Line{yz}$ includes
  all vertices in $X_Y, Y$ and $Z$. Moreover, in $X'$ it includes all
  vertices that are at distinct distances from $y$ and $z$. These are
  exactly the vertices in the symmetric difference of $N(y)$ and
  $N(z)$ in $X'$. Suppose that $F_{YZ}$ has cardinality 1, then the
  symmetric difference between $N(y)$ and $N(z)$ is the same for all
  $y$ in $Y$. This implies that all vertices of $Y$ have the same
  neighbourhood in $X'$. In other words, every vertex of $X'$ is
  adjacent to all vertices of $Y$ or to none of them. But our
  definition of $X'$ states that all its vertices have at least one
  neighbour in $Y$. We may conclude that $Y$ is complete to $X'$.  Now
  consider a vertex $x$ in $X_Y$. It has at least one neighbour $y$ in
  $Y_1$ and the line $\Line{xy}$ is universal which is a
  contradiction.
\end{proof}


\subsubsection{When both $X_Y$ and $X_Z$ are non-empty.}

Let us now suppose that $X_Y$ and $X_Z$ are non-empty. We prove that
we find many lines.

\begin{proposition}
  Given a bisplit graph $G$ such that $X_Y$ and $X_Z$ are non-empty,
  the metric space induced by $G$ satisfies the \dbelong.
  \label{prop:nonempty}
\end{proposition}

\begin{proof}
  In addition to the three families of lines described in
  facts~\ref{fact:X'},~\ref{fact:Y} and~\ref{fact:XZ}, we shall find
  three more families, namely $F_{YX_Z}$, $F_{Z_1Z}$ and
  $F_{X_YX_Z}$. First notice that since $X_Z$ is non-empty, in the
  same manner as \eqref{eq:XY} leads to \eqref{eq:Y1} and
  \eqref{eq:Y3}, we may assume that 
  \begin{equation}
    Z_1 \text{ has cardinality at least } 2 \text{ and } Z \text{ has cardinality at least } 3.
    \label{eq:Z}
  \end{equation}
  Moreover, using the same arguments as for facts~\ref{fact:Y}
  and~\ref{fact:XZ}, the set $F_{Z_1Z}$ defines ${|Z_1| \choose 2} +
  |Z_1||Z_2|$ lines which are distinct from lines in $F_{X'X'}$ and
  the set $F_{X_ZY}$ contains at least $|Y|$ lines which are distinct
  from lines in $F_{X'X'}$ or in $F_{Z_1Z}$. We may observe
  additionally that any line in these new families has a non-empty
  intersection with $X_Z$. Then they are all distinct from lines in
  $F_{Y_1Y}$. They cannot be equal to a line in $F_{X_YZ}$ since they
  either intersect $Y$ on a singleton, or $Z$ on exactly two vertices.

  In the end, we also add the lines in $F_{X_YX_Z}$. Such lines
  intersect both $X_Y$ and $X_Z$ on singletons (the generators) and
  have at least two elements in $Y_1$ and in $Z$. They are all
  distinct so they define $|X_Y||X_Z|$ lines. All of them are distinct
  from the lines described above. Table~\ref{tab:nonempty} gives a quick
  overview of the considered families and the reason why they are
  distinguished from one another.


  \begin{table}[ht]
  \begin{center}
    \scriptsize
    \begin{tabular}{|c|c||cccccc||c|}
      \hline 
      
      & & \multicolumn{6}{c||}{\rule{0pt}{.3cm} Intersection of $\Line{ab}$ with } &
      \\ 

      Family & Generators & $X'$ & $X_Y$ & \multicolumn{2}{c}{$Y_1 \cup Y_2$} & $Z_1 \cup Z_2$ & $X_Z$
      & Number of lines\\

      \hline 

      \rule[-.25cm]{0pt}{.5cm} $F_{X'X'}$ & $a \in X', b\in X'$ & $\{a,b\}$ & $\emptyset$ & & &
      & $\emptyset$ & $\displaystyle{|X'| \choose 2}$ \\

      \hline

       \rule[-.25cm]{0pt}{.5cm}  $F_{Y_1Y}$ & $a \in Y_1, b\in Y$ & & $\geq 1$ & \multicolumn{2}{c}{$\{a,b\}$} & $Z_1 \cup Z_2$ &
      $\emptyset$ & $\displaystyle{|Y_1| \choose 2} + |Y_1||Y_2|$ \\

       \hline

       \rule[-.2cm]{0pt}{.5cm}  $F_{X_YZ}$ & $a \in X_Y, b \in Z$ & & $\geq 1$ & \multicolumn{2}{c}{$Y_1 \cup Y_2$} & $\{b\}$ &
       & $ \geq |Z|$ \\

      \hline

      \rule[-.25cm]{0pt}{.5cm}  $F_{Z_1Z}$  & $a \in Z_1, b\in Z$ & & $\emptyset$ & \multicolumn{2}{c}{$Y_1 \cup Y_2$} & $\{a,b\}$ &
      $\geq 1$ & $\displaystyle{|Z_1| \choose 2} + |Z_1||Z_2|$ \\

      \hline

       \rule[-.2cm]{0pt}{.5cm}  $F_{X_ZY}$ & $a \in X_Z, b \in Y$ & & & \multicolumn{2}{c}{$\{b\}$} & $Z_1 \cup Z_2 $ &
       $\geq 1$ & $ \geq |Y|$ \\
      
      \hline
      
       \rule[-.2cm]{0pt}{.5cm}   $F_{X_YX_Z}$ & $a \in X_Y, b \in X_Z$ & & $\{a\}$ & \multicolumn{2}{c}{$\geq 2$} & $\geq 2$ &
       $\{b\}$ & $ \geq |X_Y||X_Z|$ \\
      
      \hline

    \end{tabular}
  \end{center}
  \caption{Families of lines and their intersections}
  \label{tab:nonempty}
  \end{table}


  Now we can sum all those lines. Since $Y_1$ and $Z_1$ have
  cardinality at least 2, families $F_{Y_1Y}$ and $F_{Z_1Z}$ each
  provides at least one line. Moreover, $|X'| \choose 2$ is lower
  bounded by $|X'|-1$ and $|X_Y||X_Z|$ is lower bounded by $|X_Y| +
  |X_Z| -1$. In the end we have at least
  \begin{equation*}
    |X'| + |X_Y| + |X_Z| + |Y| + |Z| \text{ lines,}
  \end{equation*}
  which is the order of the graph. This concludes the proof of
  Proposition~\ref{prop:nonempty}.
\end{proof}


\subsubsection{When $X_Y$ is non-empty and $X_Z$ is empty.}

We now consider the case when $X_Z$ is empty. In this situation, the
set $Z_1$ is also empty. Thus, the last three families of lines in
Table~\ref{tab:nonempty} cannot be used anymore. We introduce a new
family $F_{X_YY_2}$ and through a straightforward analysis of
distances (see Figure~\ref{fig:general}) we obtain
Table~\ref{tab:xzempty}. 


\begin{table}[ht]
  \begin{center}
    \scriptsize
    \begin{tabular}{|c|c||ccccc||c|}
      \hline 
      
      & & \multicolumn{5}{c||}{\rule{0pt}{.3cm} Intersection of $\Line{ab}$ with } &
      \\ 

      Family & Generators & $X'$ & $X_Y$ & $Y_1$ &  $Y_2$ & $Z$ & Number of lines\\

      \hline 

      \rule[-.25cm]{0pt}{.5cm} $F_{X'X'}$ & $a \in X', b\in X'$ & $\{a,b\}$ & $\emptyset$ & & &
        & $\displaystyle{|X'| \choose 2}$ \\

      \hline

       \rule[-.25cm]{0pt}{.5cm}  $F_{Y_1Y}$ & $a \in Y_1, b\in Y$ & & $\geq 1$ & \multicolumn{2}{c}{$\{a,b\}$} & $Z$ &
       $\displaystyle{|Y_1| \choose 2} + |Y_1||Y_2|$ \\

       \hline

       \rule[-.2cm]{0pt}{.5cm}  $F_{X_YZ}$ & $a \in X_Y, b \in Z$ & & $\geq 1$ & \multicolumn{2}{c}{$Y_1 \cup Y_2$} & $\{b\}$ &
       $ \geq |Z|$ \\

      \hline

      \rule[-.25cm]{0pt}{.5cm}  $F_{X_YY_2}$  & $a \in X_Y, b\in Y_2$ & & $\{a\}$ & $\geq 2$ & $\{b\}$ & $Z$&
      $|X_Y||Y_2|$ \\

      \hline     
      
    \end{tabular}
  \end{center}
  \caption{Families of lines  when $X_Z$ is empty}
  \label{tab:xzempty}
\end{table}


The first three families are distinguished by
facts~\ref{fact:X'},~\ref{fact:Y} and~\ref{fact:XZ}. Lines in
$F_{X_YY_2}$ are also different from lines in $F_{X'X'}$ and
$F_{Y_1Y}$. Moreover, such a line could be equal to a line in
$F_{X_YZ}$ in a very specific case only. Indeed, suppose that
$a,b,c,d$ are vertices in $X_Y, Y_2, X_Y$ and $Z$ respectively, such
that $\Line{ab} = \Line{cd}$. First note that $a$ must equal $c$
(consider intersection with $X_Y$) so $\Line{ab} =
\Line{ad}$. Moreover, in $Y_1$, the line $\Line{ab}$ contains only the
neighbours of $a$. Since line $\Line{ad}$ contains the whole set $Y,$
vertex $a$ must be complete to $Y_1$; and $Y_2$ must be a
singleton. Thus, $Y_2$ is exactly $\{b\}$. Observing the intersections
of these lines with $Z$, we need $Z$ to be the singleton $\{d\}$. All
vertices of $X'$ have a neighbour in $Z$ so $d$ is complete to
$X'$. As a consequence $\Line{ad}$ contains all vertices of $X'$ that
are at distance 3 from $a$. But since $\Line{ab}$ intersects $X'$ only
on vertices at distance 2 from $a$, this implies that no element of
$X'$ is at distance $3$ from $a$. Since $b$ has degree at least $2$
(or the graph has a bridge and thus a universal line), there must be a
vertex $x$ in $X'$ that is a neighbour of $b$. In return, vertex $x$
is in $\Line{ab}$ but it cannot be in $\Line{ac}$ since it is at
distance 2 from $a$. This yields a contradiction. Thus all four
families of lines are disjoint.


\paragraph{If $Y_2$ has two or more elements.}

Then we may sum all lines described in Table~\ref{tab:xzempty}. We get
at least
\begin{equation*}
  {|X'| \choose 2} + {|Y_1| \choose 2} + |Y_1||Y_2| + |Z| + |X_Y||Y_2|
  \text{ lines.}
\end{equation*}
Note that $x \choose 2$ is always bounded below by $x -1$ and $xy$ is
bounded below by $x+y-1$ when both $x$ and $y$ are positive integers
and by $x+y$ if both are at least 2 (recall that $|Y_1|$ is at least
2). By applying these common properties (stated in
Section~\ref{sec:calculus}), we bound the number of lines by
\begin{equation*}
  |X'| + |Y_1| + |Y_2| + |X_Y| + |Z| + (|Y_1| + |Y_2| - 3).
\end{equation*}
And by \eqref{eq:Y3}, we may conclude that graph $G$ has sufficiently
many lines.


\paragraph{If $Y_2$ is a singleton.}

In that case, our four families bring at least 
\begin{equation*}
  |X'| + |Y_1| + |Y_2| + |X_Y| + |Z| + (|Y_1| + |Y_2| - 4) \text{ lines.}
\end{equation*}
By \eqref{eq:Y3}, we only miss one line to reach our goal. 

$\bullet$ If $Z$ is a singleton $\{z\}$, then every vertex in $X'$ is adjacent
to $z$. Thus, they are all at distance 2 from one another and all the
lines in $F_{X'X'}$ contain $z$. Then $z$ is an element in all our
lines but there must be a line that does not go trough $z$ (or there
is a universal line). Then $G$ satisfies the \dbelong.

$\bullet$ If $Z$ and $X_Y$ both have size at least 2, then no line of our
families contains $X_Y, Y$ and $Z$. But we may easily consider the
line generated by the end vertices of any edge between $Y$ and $Z$ and
see that it contains all $X_Y, Y$ and $Z$. It is then different from
all considered lines and we have sufficiently many lines.

$\bullet$ If $Z$ has size at least 2 and $X_Y$ is a singleton, we consider only
the three families of lines $F_{X'X'}$, $F_{Y_1Y}$ and $F_{X_YZ}$. By
usual inequalities they provide at least
\begin{equation*}
  |X'| + |Y_1| + |Z| + (|Y_1| -2) \text{ lines.}
\end{equation*}
Since $Y_1$ has cardinality at least 2, we only need two more lines to
reach the order of $G$. Note that in our three families, no line
contains all $X_Y$, $Y$ and $Z$. By Fact~\ref{fact:YZ} we may add two
lines from $F_{YZ}$.


\paragraph{If $Y_2$ is empty.}

In this last case, we shall exhibit one last family of lines. For
this, let $y_0$ be a vertex in $Y$ with maximum degree in $X_Y$. Since
every vertex in $X_Y$ has degree at least 2 and all neighbours must be
in $Y$, we know that $y_0$ has at least $\lceil \frac{2|X_Y|}{|Y|}
\rceil$ neighbours in $X_Y$.

Now let $X_0$ denote a largest set of vertices in $X_Y$ which are at
distance 2 from each other. The set $X_0$ has size at least $\lceil
\frac{2|X_Y|}{|Y|} \rceil$. Now, lines in $F_{X_0X_0}$ intersect $X_0$
in exactly two vertices (the generators). Moreover, they do not
intersect $Z$. The families of lines considered are shown in
Table~\ref{tab:y2empty}. Those lines are all distinct except for the
two last rows if $Z$ is a singleton.

\begin{table}[ht]
  \begin{center}
    \scriptsize
    \begin{tabular}{|c|c||ccccc||c|}
      \hline 
      
      & & \multicolumn{5}{c||}{\rule{0pt}{.3cm} Intersection of $\Line{ab}$ with } &
      \\ 

      Family & Generators & $X'$ & $X_0$ & $X_Y$ &  $Y$ & $Z$ & Number of lines\\

      \hline 

      \rule[-.25cm]{0pt}{.5cm} $F_{X'X'}$ & $a \in X', b\in X'$ &
      $\{a,b\}$ & $\emptyset$ & $\emptyset$ & & & $\displaystyle{|X'|
        \choose 2}$ \\

      \hline

       \rule[-.25cm]{0pt}{.5cm}  $F_{YY}$ & $a \in Y, b\in Y$ & & & $\geq 1$ & $\{a,b\}$ & $Z$ &
       $\displaystyle{|Y| \choose 2}$ \\

       \hline

      \rule[-.25cm]{0pt}{.5cm}  $F_{X_0Y_0}$  & $a \in X_0, b\in Y_0$ & & $\{a,b\}$ &  & $\geq 1$ & $\emptyset$&
      $\displaystyle{|X_0| \choose 2}$ \\

      \hline     
      
      \rule[-.2cm]{0pt}{.5cm}  $F_{YZ}$  & $a \in Y, b\in Z$ & & & $X_Y$ & $Y$ & $Z$ & $\geq 2$ \\
      
      \hline

      \rule[-.2cm]{0pt}{.5cm}  $F_{X_YZ}$ & $a \in X_Y, b \in Z$ & & & $\geq 1$ & $Y$ & $\{b\}$ &
      $ \geq |Z|$ \\
      
      \hline
    \end{tabular}
  \end{center}
  \caption{Families of lines  when $X_Z$ and $Y_2$ are empty}  
  \label{tab:y2empty}
\end{table}

Note that since $|Y|$ is at least 3 and $|X_0| \geq  \lceil
\frac{2|X_Y|}{|Y|} \rceil$, Lemma~\ref{lem:chiant} tells us that
whenever $|X_Y|$ or $|Y|$ is not 3, we have:
\begin{equation}
  \label{eq:smart}
  {|X_0| \choose 2} + {|Y| \choose 2} \geq |X_Y| + |Y| - 1.
\end{equation}
Moreover, if $X_Y$ and $Y$ both have size exactly 3, all vertices of
$X_Y$ are at distance 2 from each other. Indeed, either a vertex of $Y$
has degree 3 in $X_Y$, or all of them must have degree 2 (they must be
incident to at least six edges coming from $X_Y$) and as a consequence
all vertices of $X_Y$ have degree 2 ($X_Y$ and $Y$ induce a cycle of
length 6). In both cases, we deduce that all vertices of $X_Y$ are at
distance $2$ from each other. Thus, we may always consider that $|X_0|$
has value $3$. Therefore~\eqref{eq:smart} remains true.

\subparagraph{When $Z$ is not a singleton.} 

If $Z$ is not a singleton, we may sum all rows of
Table~\ref{tab:y2empty}. By~\eqref{eq:smart} we have a lower bound of
\begin{equation*}
  |X'|+|Z|+|X_Y|+|Y| \text{ lines.}
\end{equation*}
Thus $G$ satisfies the \dbelong.

\subparagraph{When $Z$ is a singleton.}

If $Z$ has size 1, we do not count the last row of
Table~\ref{tab:y2empty}. Then we miss only one line. To find it, just
observe that there must be a vertex $x'$ in $X'$ by
$\eqref{eq:X'}$. This vertex has a neighbour $y$ in $Y$ which is equal
to $Y_1$ since $Y_2$ is empty. This guarantees that there is a vertex
$x$ in $X_Y$ at distance 2 from $x'$. Analyzing the distances in $G$,
we derive that line $\Line{xx'}$ intersects $X'$ exactly on $x'$,
$X_0$ on at most one vertex, and does not intersect $Z$. Thus, this
line is different from the first four rows in
Table~\ref{tab:y2empty}. This completes the number of lines to reach
the order of $G$.


\acknowledgements
\label{sec:ack}
This research was initiated at Recolles under the patronage of
AlCoLoCo and Universit\'e Clermont Auvergne. Authors are grateful to
Aline Parreau for unfruitful but nonetheless lively and interesting
discussions on the topic.

Moreover, our research is proudly supported by French ANR through
projects \textsc{Distancia} (ANR-17-CE40-0015) and \textsc{GraphEn}
(ANR-15-CE40-0009).

Giacomo Kahn is supported by the European Union's {\em Fonds
  Europ\'een de D\'eveloppement R\'egional (\textsc{feder})} program
though project AAP ressourcement S3-DIS4 (2015-2018).


\nocite{*}
\bibliographystyle{abbrvnat}
\bibliography{sample-dmtcs}
\label{sec:biblio}

\end{document}